\documentclass[a4paper,12pt]{revtex4}
\usepackage[latin1]{inputenc}
\usepackage{amsthm}
\usepackage{amsmath}
\usepackage{graphicx}
\usepackage{subfigure}
\usepackage{subfigure}
\usepackage{amsfonts}
\usepackage{hyperref}
\usepackage{xcolor}
\usepackage{comment}

\newtheorem{prop}{PROPOSITION}

\newtheorem{corol}{COROLLARY}
\newtheorem{defi}{DEFINITION}
\newtheorem{obs}{OBSERVATION}

\begin{document}
\title{Entanglement properties of multipartite\\ informationally complete quantum measurements}
\author{Jakub Czartowski}
\affiliation{Faculty of Physics, Astronomy and Computer Sciences, Jagiellonian University, Krak\'ow, Poland}
\author{Dardo Goyeneche}
\affiliation{Faculty of Physics, Astronomy and Computer Sciences, Jagiellonian University, Krak\'ow, Poland}
\affiliation{Faculty of Applied Physics and Mathematics, Technical University of Gda\'{n}sk, 80-233 Gda\'{n}sk, Poland}
\affiliation{Departamento de F\'{i}sica, Facultad de Ciencias B\'{a}sicas, Universidad de Antofagasta, Casilla 170, Antofagasta, Chile}
\author{Karol {\.Z}yczkowski}
\affiliation{Institute of Physics, Jagiellonian University, Krak\'ow, Poland}
\affiliation{Center for Theoretical Physics, Polish Academy of Sciences, Warsaw, Poland}

\date{May 28, 2018}

\begin{abstract}
We analyze tight informationally complete measurements for arbitrarily large multipartite systems and study their configurations of entanglement. We demonstrate that tight measurements cannot be exclusively composed neither of fully separable nor maximally entangled states. We establish an upper bound on the maximal number of fully separable states allowed by tight measurements and investigate the distinguished case, in which every measurement operator 
carries the same amount of entanglement. Furthermore, we introduce the notion of nested tight measurements, i.e. multipartite tight informationally complete measurements such that every reduction to a certain number of parties induces a lower dimensional tight measurement, proving that they exist for any number of parties and internal levels.



\end{abstract}
\maketitle

\section{Introduction}
Information stored in physical systems can be retrieved by choosing a suitable set of measurements. In quantum mechanics this information can be obtained by considering \emph{positive operator valued measures} (POVM), also known as \emph{generalized quantum measurements}. By implementing a set of measurements over an ensemble of particles it is possible to tomographically reconstruct quantum states. The information gained from some special sets of measurements is enough to reconstruct any state of a quantum system. Such kind of measurements are called \emph{informationally complete} (IC)-POVM and the process to reconstruct the state is called \emph{quantum state tomography}. These general measurements have a fundamental importance in quantum information theory, as they find applications in quantum processes tomography \cite{S08}, quantum cryptography \cite{R05} and quantum fingerprinting \cite{SWS07}.

Among existing schemes of generalised quantum measurements one distinguishes a class called \emph{tight IC-POVM} \cite{S06, RS07}, having some remarkable properties: \emph{(i)} they provide a simple tomographical reconstruction formula, \emph{(ii)} they minimize the mean square error (MSE), i.e., minimal mean square Hilbert-Schmidt distance between the estimate and the true state under the presence of errors in state preparation and measurements \cite{S06}. Tight IC-POVM, that we also call tight measurement, are closely related to \emph{complex projective t-designs} \cite{N81} for $t\geq 2$. These measurements are optimal for quantum clonning \cite{S06} and have been widely used for theoretical \cite{GA16,Z99,KR05,DCEL09,GAE07,AE07,CLM15,XZ16} and experimental \cite{LNGGNDVS11,ECGNSXL13,B15,DBBV17} applications in quantum information theory. They have recently been applied in a protocol for teleportation of quantum entanglement over 143 km \cite{HSFHWUZ15}.

Along this paper, we study tight measurements for multipartite quantum systems and reveal several aspects concerning amount of entanglement allowed by these configurations. In particular, we shed some light onto the following basic question:
\begin{center}
 \emph{For what multipartite quantum systems one can find an isoentangled tight IC-POVM?}
\end{center}
After the introductory Section \ref{S2}, we start Section \ref{S3} by showing that the answer to the above question is not trivial: a tight measurement cannot be composed by fully separable measurement operators. As a matter of fact, the answer depends on which measure of quantum resources we choose to compare two given quantum measurements. We will optimize quantum resources by considering \emph{(i)} tight measurements composed by the maximal possible number of measurement operators that can be prepared by considering local operations and classical communication (LOCC), and \emph{(ii)} isoentangled tight measurements, in the sense that all measurement operators can be defined through a single fiducial state and local unitary operations.\bigskip

This work is organized as follows: In Section \ref{S2}, we recall the formal definition of tight measurements and study some basic properties. In Section \ref{S3} we study the distribution of entanglement in such quantum measurements and show that they cannot be composed exclusively neither of fully separable nor $k$-uniform states, i.e. pure states such that every reduction to $k$ parties is maximally mixed. We also show that Symmetric IC (SIC)-POVM cannot be composed by grouping fully separable and $k$-uniform states for any $N$ qudit systems and any $k\geq1$, establishing a remarkable difference with respect to the mutually unbiased bases (MUB), where we show that this is possible for every $N$ qudit system and $k=1$ only. In Section \ref{S4} we introduce the notion of nested tight measurements and prove that they exist for any $N$-qudit system. In Section V, we present a summary of our main results and conclude the work. Appendix \ref{Ap1} contains a proof of propositions, whereas Appendix \ref{Ap2} shows a particular configuration of a two-qubit SIC-POVM containing five separable states, which is the maximal number we managed to find.

\section{Tight measurements}\label{S2}
In this section we review the concept of \emph{tight informationally complete-POVM} and show some basic properties and examples. Let us start by recalling the definition of \emph{Positive Operator Valued Measure} (POVM). A POVM is a set of $m$ positive semidefinite operators $\Pi_j$ summing up to the identity. Along the work, we restrict our attention to subnormalized rank-one POVM, i.e., operators of the form
\begin{equation}\label{Pi}
\Pi_j=\frac{D}{m}\,|\varphi_j\rangle\langle\varphi_j|,\hspace{0.3cm}j\in\{1,\dots,m\},
\end{equation}
where $|\varphi_j\rangle\in\mathbb{C}^D$ and $m\geq D$, such that $\sum_{j=1}^m\Pi_j=\mathbb{I}$. Here $m$ denotes the number of outcomes of the POVM. \emph{Weighted frame potentials} \cite{BF03} can be used to quantify how far a given generalized measurement  is from the set of complex projective $t$-designs. The weighted frame potential of order $t$ of a given POVM $\{\Pi_j\}$ with $m$ outcomes reads \cite{S06}:
\begin{equation}\label{potential}
F_t=\sum_{i,j=1}^{m}\omega_i\omega_j\bigl(\mathrm{Tr}(\Pi_i\Pi_j)\bigr)^t,
\end{equation}
where $\{\omega_j\}_{j=1,\dots,m}$ is a real and normalized weighted function, i.e. $0\leq \omega_j\leq1$ and $\sum_{j=1}^{m}\omega_j =1$. For the sake of simplicity, we are going to introduce tight measurements through the notion of weighted frame potentials. The following adopted definition is a simplification of Corollary 14 in Ref.\cite{S06}: 
\begin{defi}
A set of $m$ subnormalized rank-one projectors $\{\Pi_j\}$ acting on a Hilbert space of dimension $D$ is called a tight IC-POVM if for $t=2$ the weighted frame potential $F_t$ given in Eq.(\ref{potential}) saturates the following lower bound
\begin{equation}\label{Welch}
F_t\geq \frac{D^{t}}{m^{t-2}}\binom{D+t-1}{t}^{-1}.
\end{equation}
\end{defi}

Eq.(\ref{Welch}) is also known as \emph{generalised Welch bound} \cite{B08}. A set of vectors saturating Eq.(\ref{Welch}) for a given $t\geq2$ defines a \emph{weighted complex projective $t$-design} \cite{S06}. The factor $(m/D)^t$ appearing in Eq.(\ref{Welch}) is due to the fact that the projectors (\ref{Pi}) are subnormalized. Note that saturation of bound (\ref{Welch}) for $t=1$ occurs for any rank-one POVM, equivalently, complex projective 1-design or tight frame \cite{W03}. For instance, it is saturated by \emph{von Neumann projective measurements}, i.e. orthonormal bases, for which $m=D$.

A remarkable example of tight measurements is formed by 
\emph{Symmetric IC-POVM} \cite{Z99}, for  which $m=D^2$, the weight function reads 
$\omega_i=1/m$ for $i=1,\dots,m$, and all projectors have the same pairwise overlap with respect to
the Hilbert-Schmidt product. These configurations 
form a projective $2$-design 
\cite{RBSC04}, 
they exists in every dimension $d\leq21$ \cite{Z99,RBSC04,Marcus,Markus1, ACFW,GS17} and some other dimensions up to $124$ \cite{GS17}. Furthermore, highly precise numerical solutions 
were obtained for every dimension $d\leq 151$ 
 \cite{Z99,RBSC04,Scott, Andrew, FHS,GS17}
and also  $d=323$ and $d=844$ \cite{GS17}.

Another highly interesting class of tight measurement is represented by maximal sets of mutually unbiased bases (MUB) \cite{I81}, for which $m=D(D+1)$ and the weight function reads 
$\omega_i=1/m$ for  $i=1,\dots,m$. These configurations 
form a projective $2$-design in every dimension in which they exist \cite{KR05}. 
Two orthonormal bases $|\phi_i\rangle$ and $|\psi_j\rangle$ defined in dimension $D$ are \emph{unbiased} if $|\langle\phi_i|\psi_j\rangle|^2=1/D$, for every $i,j=1,\dots,D$. A set of more than two orthonormal bases form \emph{mutually unbiased bases} (MUB) if they are pairwise unbiased. It is known that at most $D+1$ MUBs exist in every dimension $D$, and 
the bound is saturated for any dimension $D$ which can be factored as a power of a prime $d$, $D=d^N$ \cite{WF89}. Some more general classes of tight measurements have been constructed in Ref. \cite{RS07}. \medskip

The existence and construction of tight measurements forms a challenging open problem. Indeed, the maximal number of MUBs in dimension $D=6$ is still unknown, while SIC-POVMs are currently 
known to exist in a finite set of dimensions only. A remarkable property of tight measurements $\{\Pi_j\}$ is the fact that any $D$ dimensional quantum state $\rho$ can be decomposed \cite{S06} 
as follows:
\begin{equation}\label{rho}
    \rho = \frac{m(D + 1)}{D}  \sum_{j=1}^m p_j \Pi_j - \mathbb{I}_{D},
\end{equation}
where $p_j=\mathrm{Tr}(\rho\Pi_j)$ is a normalized probability distribution that can be measured in the laboratory. Furthermore, a tight IC-POVM offers a high fidelity of state 
reconstruction under the presence of errors in state preparation and quantum measurements \cite{S06}.\medskip

Let us show that probability distribution $\{p_j\}$ from Eq.(\ref{rho}) satisfies a special constraint.
\begin{prop}\label{prop1}
Let $\{\Pi_j\}$ be a tight measurement acting on dimension $D$ and having $m$ outcomes. Then, the following relation holds,
\begin{equation}\label{p2j}
\sum_{j=1}^m p^2_j=\frac{D(\mathrm{Tr}(\rho^2)+1)}{m(D+1)}.
\end{equation}
\begin{proof}
Multiplying Eq.(\ref{rho}) by $\rho$ and taking the trace at both sides we arrive 
at the desired result.
\end{proof}
\end{prop}
In particular, Proposition \ref{prop1} holds for any maximal set of MUBs and SIC-POVMs. This proposition induces a sort of Bloch representation for any tight measurement. In particular, for a single-qubit system and maximal set of MUB measurements, three out of six probabilities are independent. By applying a suitable redefinition of these independent coordinates we obtain the standard Bloch representation of a one-qubit state.

One might be tempted to ask for  the reason to restrict our attention to rank-one projective measurements, as multipartite states of an arbitrary rank represent the general case. However, there is a much better understanding of multipartite entanglement for pure states, compared to multipartite mixed states. For instance, already for the two-qutrit system no general method to decide whether a given density matrix of order nine represents an entangled state is known \cite{S16}.  On the other hand, in this work we take advantage of complex projective $t$-designs, which are defined for pure states and naturally induce rank-one tight informationally complete generalized quantum measurements \cite{S06}. A further research propose a generalization of complex projective $t$-designs to mixed quantum states \cite{CGGZ18}.

\section{Entanglement in tight measurements}\label{S3}
In this section we study entanglement properties of multipartite tight measurements. We derive some general results that do not depend on global unitary transformations applied to the entire set of measurement operators. Let us start by showing a simple fact.
\begin{obs}\label{obs1}
An informationally complete set of measurements can be exclusively formed by local operators.
\end{obs}
For instance, in the case of an $N$--qudit system,
described by a Hilbert space of dimension $D=d^N$,
 the fully separable eigenbases of the tensor product of $d$-dimensional generalized Pauli group allows us to  univocally determine any density matrix of size $D$. This way to reconstruct of quantum states is known as \emph{standard quantum tomography}, and it is highly costly because the number of projective measurements scales as $D^3=d^{3N}$, for $N$ qudit systems \cite{L97}. On the other hand, let us show that a tight measurement cannot be constructed out of local operators.
\begin{prop}\label{prop2}
A tight measurement cannot be exclusively formed by local operators.
\end{prop} 
\begin{proof}
Suppose that $\{\Pi_j\}$ is a bipartite tight IC-POVM such that $\Pi_j=\Pi^A_j\otimes\Pi^B_j$, where $\{\Pi^A_j\}$ and $\{\Pi^B_j\}$ are tight measurements having $m_A$ and $m_B$ outcomes and acting on dimensions $D_A$ and $D_B$, respectively. From Eq.(\ref{rho}) we have
\begin{equation}\label{tensorIC}
    \rho = (D_AD_B + 1) \frac{m_Am_B}{D_AD_B} \sum_{j=1}^m p_j \Pi^A_j\otimes\Pi^B_j - \mathbb{I}_{D_AD_B}.
\end{equation}
Let us consider the separable state $\rho=\rho_A\otimes\mathbb{I}_{D_B}$ and the identities 
\begin{equation}
p_j=\mathrm{Tr}[((\rho_A)\otimes\mathbb{I}_{D_B})(\Pi^A_j\otimes\Pi^B_j)]=\frac{p^A_j}{m_B},
\end{equation}
where $p^A_j=\mathrm{Tr}(\rho_A\Pi^A_j)$ and
\begin{equation}\label{identity}
\sum_{j=1}^{m} p_j \Pi^A_j=m_B\sum_{j=1}^{m_A} p_j \Pi^A_j.
\end{equation}
Identity (\ref{identity}) holds because of the $m=m_Am_B$ operators $\Pi^A_j\otimes\Pi^B_k$ cover all possible combination of indices $j=1,\dots,m_A$ and $k=1,\dots,m_N$. 
Thus, taking partial trace over the subsystem $B $ in Eq.(\ref{tensorIC}) we obtain
\begin{equation}\label{rhoA}
    \rho_A = \frac{m_A(D_AD_B + 1)}{D_A} \sum_{j=1}^{m_A} p_j \Pi^A_j - D_B\,\mathbb{I}_{D_A}.
\end{equation}
From here we find a contradiction, as the reduced measurement $\{\Pi^A_j\}$ would not be tight for any dimension $D_B\geq2$. That is, Eq.(\ref{rho}) does not hold for any $D_B\geq2$, which concludes the proof. For multipartite systems the proof follows in the same way by considering every possible bipartition.
\end{proof}
Tensor product of tight measurements defines an informationally complete set of measurements, but not a tight one. For instance, the set of $N$ partite measurements composed by product of monopartite SIC-POVM is informationally complete and optimal among all product measurements \cite{ZE11}. However, the fact that the resulting multipartite measurement is not tight has important consequences: \emph{robustness of fidelity reconstruction under the presence of measurement errors decreases exponentially with the number of parties, with respect to the tight SIC-POVM} \cite{ZE11}. 

In Section \ref{S4}, we show that some products of tight measurements can be complemented with POVM composed of maximally entangled states in such a way that the entire set forms a tight measurement.\medskip

Let us now consider further classes of multipartite entangled states. A multipartite pure quantum state is called $k$-\emph{uniform} if every reduction to $k$ parties is maximally mixed \cite{S04,AC13}. Let us present the following result.
\begin{prop}\label{prop3}
Tight measurements cannot be exclusively composed of $k$-uniform states for any $k\geq1$.
\end{prop}
\begin{proof}
Suppose there is a tight IC-POVM such that every measurement operator is a $k$-uniform state. Therefore, reconstruction formula (\ref{rho}) does not work for states $\rho$ having non-maximally mixed reductions.
\end{proof}
 
Propositions \ref{prop2} and \ref{prop3} reveal that for a tight measurement only intermediate degrees of average entanglement are allowed, where the average is taken over all measurement operators forming the POVM. There is a simple argument to estimate the mean amount of entanglement characterising the quantum states which lead to a tight measurement. Consider a quantum system composed by $N$ qudits having $d$ levels each, where the total dimension of the Hilbert space is $D=d^N$. Also consider reductions $\rho_{X_i}=Tr_{\overline{X}_i}(\rho)$, where $X_i$ denotes the $i$-th subset of $k$ out of $N$ parties, $\overline{X}_i$ is the complementary set and $i=1,\dots,\binom{N}{k}$. For instance, for $N=4$ and $k=2$ we have $X_1=\{1,2\}$, $X_2=\{1,3\}$, $X_3=\{1,4\}$, $X_4=\overline{X}_3=\{2,3\}$, $X_5=\overline{X}_2=\{2,4\}$, $X_6=\overline{X}_1=\{3,4\}$. A tight measurement is a 2-design, which implies that the average purity of reductions over a tight measurement coincides with average purity over the entire set of quantum states, according to the Haar measure distribution. That is,
\begin{equation}\label{Haar}
\frac{1}{m}\sum_{j=1}^m\mathrm{Tr}({\sigma_j}^2_{X_i})=\langle\mathrm{Tr}(\rho^2_{X_i})\rangle_{\mathrm{Haar}},
\end{equation}
where $\sigma_j=\frac{m}{D}\Pi_j=|\varphi_j\rangle\langle\varphi_j|$ are normalised projectors and $\rho_{X_i}$ is a $k$-qudit reduction of the state $\rho$, with respect to the subset of parties $X_i$. Using the average values derived by Lubkin \cite{L78}, we have
\begin{equation}\label{average}
\frac{1}{m}\sum_{j=1}^m\mathrm{Tr}({\sigma_j^2}_{X_i})=\frac{d^k+d^{N-k}}{d^N+1},
\end{equation}
for every reduction to $k$ parties $X_i$. In the particular case of maximal sets of MUB for bipartite systems this expression reduces to Eqs.(3-5) in Ref.  \cite{WPZ11}. From Eq.(\ref{average}), we see that the average purity of reductions of a tight measurement behaves asymptotically as
\begin{equation}\label{asympt}
\frac{1}{m}\sum_{j=1}^m\mathrm{Tr}({\sigma_j^2}_{X_i})\rightarrow\frac{1}{d^k}.
\end{equation}
This statement holds if either the number of levels $d$ or the number of parties $N$ is large, for any $k\leq N/2$. Therefore, projectors $\sigma_j$ forming a tight measurement asymptotically become close to \emph{absolutely maximally entangled} (AME) states in large Hilbert spaces. We recall that AME states for $N$ qudit systems are $k$-uniform states for the maximal possible value $k=\lfloor N/2\rfloor$ \cite{HCRLL12}. For instance, the two-qudit generalized Bell states and three-qubit GHZ states belong to this class. 

Despite the  asymptotic behavior described above
it is always possible to apply a rigid rotation to a tight IC-POVM in such a way 
that some measurement operators become fully separable, i.e. they form local measurements. 
An interesting question concerns establishing how many of the measurement operators can be chosen to be fully separable. From Eq.(\ref{average}) we can estimate this maximal possible number $m_{sep}$ of fully separable operators, i.e. those satisfying the constraint $\mathrm{Tr}({\rho_j^2}_{X_i})=1$, $j=1,\dots,m_{sep}$, for every subset $X_i$. To this end, we impose the extremal condition that the remaining $m-m_{sep}$ operators are $k$-uniform states, i.e. $\mathrm{Tr}({\sigma_j}^2_{X_i})=1/d^k$, for every subset $X_i$ and every $j=m_{sep}+1,\dots,m$. As a consequence, we arrive at the following statement.
\begin{prop}\label{prop4}
Suppose that a tight measurement for $N$ qudit systems is composed by $m_{sep}$ fully separable operators. Also, suppose that among the $m-m_{sep}$ remaining operators  there is at least a $k$-uniform state, where $k$ denotes the maximal possible uniformity. Then, the following relation holds:
\begin{equation}\label{boundm_{sep}}
m_{sep}\,(d^N+1)\leq m\,(d^k+1).
\end{equation}
This inequality is saturated if and only if the remaining $m-m_{sep}$ states are $k$-uniform, for a given $k\geq1$.
\end{prop}
\begin{proof}
Consider Eq.(\ref{average}), where $m_{sep}$ reductions are pure states and $m-m_{sep}$ reductions are maximally mixed for every possible subset $X_i$ consisting of $k$ out of $N$ parties. Therefore we have
 \begin{equation}\label{sep-ent}
\frac{1}{m}[m_{sep}+d^{-k}(m-m_{sep})]=\frac{d^k+d^{N-k}}{d^N+1}.
\end{equation}
From here we obtain Eq.(\ref{boundm_{sep}}) as an equality. The inequality occurs because the possible number of separable vectors $m_{sep}$ is strictly lower when not all of the $m-m_{sep}$ entangled states are $k$-uniform.
\end{proof}
For example, for two-qubit systems ($d=N=2$), a tight measurement composed by at least one Bell state ($k=1$) satisfies $m_{sep}/m\leq3/5$. Note that this inequality is saturated by a maximal set of 5 MUB composed by 3 separable and 2 maximally entangled bases ($m_{sep}/m=12/20=3/5$). In general, Proposition \ref{prop4} implies that a maximal set of $D+1$ mutually unbiased bases existing in dimension $D=d^N$ cannot contain more than $d^k+1$ fully separable bases, where $k\leq\lfloor N/2\rfloor$. Moreover, from the fact that at most $d+1$ fully separable MUBs exist for $N$ qudit systems \cite{L11},  we arrive to the following result.
\begin{corol}\label{corol1}
A maximal set of $d^N+1$ MUBs for an $N$--qudit system composed exclusively
of fully separable and $k$-uniform states is only possible for $k=1$. 
\end{corol}
For instance, Corollary \ref{corol1} forbids existence of a maximal set of 5-qubit MUB exclusively composed by fully separable and absolutely maximally entangled states \cite{HCRLL12}, which are $k=2$ uniform. There is another interesting consequence of Proposition \ref{prop4}. For a fixed number of parties $N$ and reductions $k\leq N/2$, left hand side of Eq.(\ref{boundm_{sep}}) increases faster than the right hand side, as a function of local dimension $d$. As a consequence, there is an upper bound on $d$ for the existence of a tight measurement exclusively composed by the union set of fully separable and $k$-uniform states.
\begin{corol}\label{corol2}
A tight measurement composed by $m$ rank-one projectors, where $m_{sep}\leq m$ are fully separable, can only exist for local dimensions $d\leq d_{max}$, where $d_{max}$ is implicitly defined as
\begin{equation}
m_{sep}\,(d_{max}^N+1)=m\,(d_{max}^{N/2}+1).
\end{equation}
\end{corol}
\begin{proof}
Consider Eq.(\ref{boundm_{sep}}) from Proposition \ref{prop4} and the maximal possible value for $k=N/2$. In general, this assumption is an overestimation (e.g. for $N=4$ and $d=2$ the value $k=2$ is not possible \cite{HS00}).
\end{proof}
 The only examples saturating the value $d=d_{max}$ established by Corollary \ref{corol2}
 is achieved by the two-qudit MUB for prime power values of $d$, as far as we know.\medskip
 
Let us now consider an implication for SIC-POVMs.
\begin{corol}\label{corol3}
For any $k\geq1$, any $N$-qudit SIC-POVM cannot be composed of 
fully separable and $k$-uniform states.
\end{corol}
\begin{proof}
Inequality $m_{sep}\leq\frac{(1+d^k)d^{2N}}{1+d^N}$ comes from considering $m=d^2$ in Eq.(\ref{boundm_{sep}}). According to Proposition \ref{prop4} we require saturation of the inequality in order to have a POVM exclusively composed by $m_{sep}$ fully separable and $m-m_{sep}$ $k$-uniform elements. However, the inequality cannot be saturated as the upper bound is not an integer number for any value of the involved parameters.
\end{proof}
Corollary \ref{corol3} implies that at most nine separable vectors are allowed by a two-qubit SIC-POVM. However, by considering exhaustive numerical optimization of 2-nd weighted frame potential (\ref{potential}), we have found no more than \emph{five} separable vectors (see Appendix \ref{Ap2}).  Corollaries \ref{corol1} and \ref{corol3} reveal fundamental differences existing between MUBs and SIC-POVMs from the point of view of quantum entanglement.\medskip

Before ending this section, let us emphasise that isoentangled tight measurement exist for multipartite systems. The most remarkable example is provided by the so-called \emph{Hoggar lines} \cite{H98}, a special kind of SIC-POVM existing for three-qubit systems. They are given by the following 64 states:
\begin{equation}\label{sssphi}
|\phi_{ijk}\rangle=\sigma_i\otimes\sigma_j\otimes\sigma_k |\phi_{000}\rangle,
\end{equation}
where $\sigma_0=\mathbb{I}_2$, while $\sigma_1$, $\sigma_2$ and $\sigma_3$ are the Pauli matrices. Hoggar lines are isoentangled, as all its 64 elements $|\phi_{ijk}\rangle$ are related by local unitary operations through Eq.(\ref{sssphi}). Among the entire set of 240 Hoggar lines \cite{Z12}, there are only two inequivalent fiducial states under local unitary operators, which can be written in the following symmetric form \cite{G17}:
\begin{eqnarray}
|\phi^{(1)}_{000}\rangle&=&\frac{1}{\sqrt{6}}\bigl(|000\rangle+i(|011\rangle+|101\rangle+|110\rangle)-(1-i)|011\rangle\bigr),\nonumber\\
|\phi^{(2)}_{000}\rangle&=&\frac{1}{\sqrt{6}}\bigl(|000\rangle+i(|011\rangle+|101\rangle+|110\rangle)-(1+i)|011\rangle\bigr).
\end{eqnarray}
A symmetric fiducial state equivalent to $|\phi^{(2)}_{000}\rangle$ was already noted by Jedwab  and  Wiebe \cite{JW15}. Symmetric fiducial state for Hoggar lines is unique up to unitary or anti-unitary transformations. However, from the point of view of quantum information theory, fiducial states $|\phi^{(1)}_{000}\rangle$ and $|\phi^{(2)}_{000}\rangle$ are essentially different, as they have different values of the three-tangle $\tau_{ABC}$, a three-qubit entanglement invariant \cite{CKW00}. These values are given by $\tau^{(1)}_{ABC}=2/3$ and $\tau^{(2)}_{ABC}=2/9$, respectively. Even though these states are inequivalent under local unitary operations, the state $|\phi^{(1)}_{000}\rangle$ can be obtained from $|\phi^{(2)}_{000}\rangle$ by considering stochastic local operations, as both states belong to the GHZ class of states characterized by the restriction $\tau_{ABC}>0$.

For bipartite systems, any existing tight measurement has an average amount of entanglement -- quantified by the purity of reductions -- that only depends on the total dimension of the Hilbert space and the dimension of the reduced space, as we have seen in Eq. (\ref{Haar}). We observe that the same result \emph{cannot be generalized} to multipartite systems.
\begin{prop}\label{prop5}
For three-partite systems, the average purity of any single particle reductions of a tight measurement is not uniquly  determined by the dimension of the Hilbert space D and the dimension of the reduction.
\end{prop}
In other words, the law of Lubkin formulated in \cite{L78} for bipartite systems (see Eq. (\ref{average})) cannot be generalized in the above sense to  three-partite systems. However, it is worth to mention that the Lubkin's law applies to any bipartition of a multipartite system.

One can pose the question whether similar isoentangled set  can lead to a set o MUBs
 acting on two-qudit systems, i.e $D = d^2$. 
As we have seen in Eq.(\ref{average}), the average entanglement among all the states of a maximal set of MUB is fixed by the size of the system and the reductions. However, by applying global unitary operations one may obtain different entanglement distributions having the same average entanglement. In particular, isoentangled sets are remarkably interesting, as they can be prepared in the laboratory by considering a single \emph{fiducial} state and local unitary operations. In the simplest case of two qubits 
the answer has recently been shown to be positive
as a configuration of five isoentangled MUBs in dimension 
$D=4$ exists \cite{CGGZ18}.

\section{Nested tight measurements}\label{S4}
In Section \ref{S3}, we have shown that tight measurements cannot be exclusively composed neither of fully separable nor of $k$-uniform states. Here, we study the possibility to construct tight measurements composed by grouping fully separable and $k$-uniform states. Let us start by showing the following result.
\begin{prop}\label{prop6}
Let $\{\Pi_j\}$ be a tight measurement defined for $N$ qudit systems and composed by $m$ rank-one projectors, where $m_{sep}$ of them are fully separable and $m-m_{sep}$ are $k$-uniform states, for a given $k\geq1$. Every reduction to $k$ parties of the entire set of projectors induce a tight measurement for $k$ qudit systems if and only if the bound (\ref{boundm_{sep}}) is saturated.
\end{prop}
A proof of Proposition \ref{prop6} is given in Appendix \ref{Ap1}. From the saturation of Eq.(\ref{boundm_{sep}}) we realize that if $(d^N+1)$ is not divisible by $(d^k+1)$ then $m_{sep}=d^k+1$, which represents the maximal number of fully separable vectors allowed by a maximal set of MUB for $N$ qudit systems. Note that Proposition \ref{prop6} generalizes the result found in Proposition \ref{prop4}.

 Let us define the main ingredient of this section.
\begin{defi}
A nested tight measurement of degree $k$ is a tight measurement such that every set of reductions to $k$ parties defines a tight measurement.
\end{defi}
In order to illustrate nested tight IC-POVM of degree $k=1$ let us consider a maximal set of 5 MUB for two-qubit systems, given by the columns of the following unitary matrices  \cite{BBRV02}:
\begin{eqnarray*}
&\mathcal{B}_1=\left(\begin{array}{cccc}
1&1&1&1\\
1&\bar{1}&1&\bar{1}\\
1&\bar{1}&\bar{1}&1\\
1&1&\bar{1}&\bar{1}
\end{array}\right)\qquad
\mathcal{B}_2=\left(\begin{array}{cccc}
1&1&1&1\\
i&\bar{i}&i&\bar{i}\\
i&\bar{i}&\bar{i}&i\\
\bar{1}&\bar{1}&1&1
\end{array}\right)\qquad
\mathcal{B}_3=\left(\begin{array}{cccc}
1&1&1&1\\
1&\bar{1}&1&\bar{1}\\
\bar{i}&i&i&\bar{i}\\
i&i&\bar{i}&\bar{i}
\end{array}\right)\qquad
\mathcal{B}_4=\left(\begin{array}{cccc}
1&1&1&1\\
\bar{i}&i&i&\bar{i}\\
1&\bar{1}&1&\bar{1}\\
i&i&\bar{i}&\bar{i}
\end{array}\right),&
\end{eqnarray*}
where the upper bar denotes negative sign and a normalization factor $1/2$ has to be applied to every vector. The remaining unbiased basis is the computational basis $\mathcal{B}_0=\mathbb{I}_4$. Note that $\mathcal{B}_0, \mathcal{B}_1$ and $\mathcal{B}_2$ are separable bases, whereas $\mathcal{B}_3$ and $\mathcal{B}_4$ are maximally entangled bases. Both Alice and Bob reductions of the three separable bases yield a maximal set of 3 MUB for the single qubit system, whereas the remaining reductions are maximally mixed.

Following the main result of Ref. \cite{L11}, we have the following observation.
\begin{obs}\label{obs2}
Nested tight measurements exist for any $N$--qudit system.
\end{obs}
This observation relies on the fact that there is a maximal set of $d^N+1$ MUB in prime power dimension $D=d^N$ such that every single particle reduction determines a maximal set of $d+1$ MUB in every prime dimension $d$ \cite{L11}. On the other hand, we have the following consequence of Proposition \ref{prop6}.
\begin{obs}\label{obs3}
Nested SIC-POVM do not exist for any $N$--qudit system.
\end{obs}
In other words, Eq.(\ref{boundm_{sep}}) cannot be saturated for a tight IC-POVM having the minimal possible number of measurement outcomes $m=d^{2N}$, for any number of parties $N$, single particle levels $d$ and reductions to $k$ parties.

To conclude the section, let us mention that Figure \ref{Fig1} summarizes the most important results of the paper. Namely, we show how robustness of informational completeness changes as a function of average entropy of single-particle reductions. Here, robustness of informational completeness is characterized by the maximal tolerance of noise under which a set of noisy measurements is able to efficiently reconstruct any quantum state \cite{S06,ZE11}.
Furthermore, the averaged entropy of a single-particle reduced density matrix
quantifies the entanglement of the analyzed multipartite pure quantum state.

\section{Concluding Remarks}\label{S5}
We studied entanglement configurations allowed by tight informationally complete quantum measurements having any number of outcomes for arbitrary large multipartite quantum systems. These informationally complete measurements, MUB and SIC-POVM in particular, have the advantage to maximize the robustness of  fidelity reconstruction under the presence of errors in both state preparation and measurement stages. We focused our study on the existence of tight measurements having not only the highest possible robustness but also allowing the simplest possible implementation, in the sense of minimizing the physical resources required to implement the measurements in the laboratory. We have shown that tight quantum measurements cannot be exclusively composed neither of fully separable nor of $k$-uniform states, for any $k\geq1$ (see Propositions \ref{prop2} and \ref{prop3}). In particular, the result holds for multipartite GHZ qudit states ($k=1$) and absolutely maximally entangled states ($k=\lfloor N/2\rfloor$).

\begin{figure}[htbp]
\includegraphics[width=12cm]{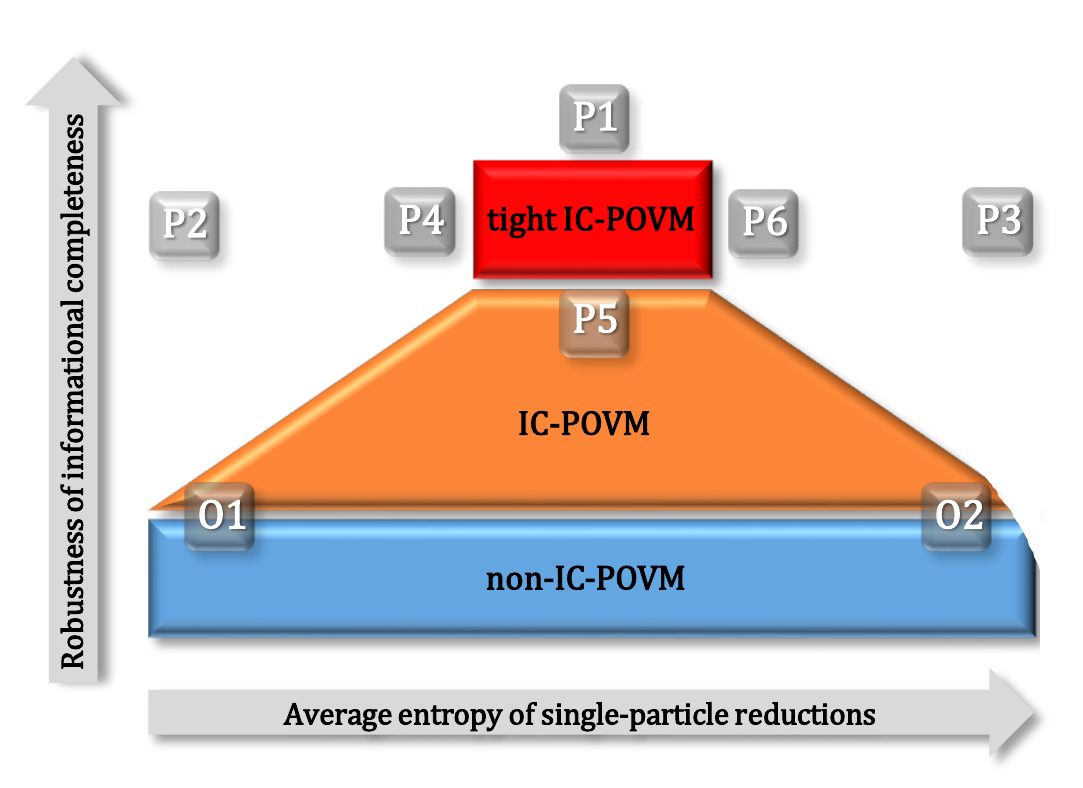}
\caption{Robustness of informational completeness of generalised quantum measurement as a function of average entropy of single-particle reductions. Strength of informational completeness is characterized by its robustness of fidelity reconstruction under noisy state preparation and imperfect measurements. Observations (O) and Propositions (P) derived along the work are illustrated (see details in Conclusions). Note that the lower-right corner of IC-POVM measurements is removed due to the fact that IC-POVM cannot be composed by maximally entangled measurements, whereas some non-IC-POVM can be maximally entangled (e.g. maximally entangled bases). Furthermore, Eq.(\ref{average}) shows that tight measurement establish the strongest possible restriction on the possible entropies of reductions. That is, only a fixed amount of entropy is allowed. The middle region IC-POVM (orange in the online version) represents an interpolation between non-IC-POVM and the strongly constrained tight IC-POVM.}.
\label{Fig1}
\end{figure}

It is also worth to mention that non-tight measurements can be actually composed by fully separable states -- see Observation \ref{obs1} and Fig. \ref{Fig1}. Furthermore, we established an upper bound for the maximal number of fully separable states that can form a part of a tight measurement. 
We showed that a tight measurement composed exclusively out of fully separable states
and $k$-uniform states 
can exist only for $k=1$ -- see Corollary \ref{corol1}. Here, a remarkable example is provided by some special maximal sets of mutually unbiased bases existing in a prime power dimension. 
An analogous statement does not hold for SIC-POVMs for any number of particles and any local dimension, establishing a fundamental difference with respect to mutually unbiased bases --
see Corollary \ref{corol3}.

Furthermore we showed that the average amount of entropy of reductions of a tripartite tight measurement, 
is not solely determined by the number of outcomes, number of parties and the dimensions of the Hilbert spaces, as is the case for the result of Lubkin (\ref{average})
valid for bipartite systems --
see Proposition \ref{prop5}. We introduced the notion of \emph{nested tight measurement}, i.e. tight measurement such that every subset of reductions to a certain number of parties induces a tight measurement (see Section \ref{S4}). These nested measurements exist for $N$ qudit system for any number of parties $N\geq2$ and any number of internal levels $d\geq2$ --
see Observation \ref{obs2}. Finally, we provided necessary and sufficient conditions to have a tight measurement for an $N$--qudit system composed by both fully separable 
and $k$-uniform states -- see Proposition \ref{prop6}.

\section*{Acknowledgements}

We are specially thankful to Marcus Appleby, Ingemar Bengtsson, Markus Grassl, Zbigniew Pucha\l a and Marcin Wie\'{s}niak for fruitful discussions about average moments of reduced density matrices, Hoggar lines and informationally complete measurements in general. Financial support by Narodowe Centrum Nauki under the grant number DEC-2015/18/A/ST2/00274 is gratefully acknowledged.

\appendix

\section{Proof of Propositions}\label{Ap1}
In this Appendix, we provide a proof of Proposition \ref{prop6}.\medskip

\noindent\textbf{PROPOSITION 6.} \emph{Let $\{\Pi_j\}$ be a tight measurement defined for $N$ qudit systems and composed by $m$ rank-one projectors, where $m_{sep}$ of them are fully separable and $m-m_{sep}$ are $k$-uniform states. Every reduction to $k$ parties of the fully separable projectors forms a tight measurement for $k$ qudit systems if and only if the bound (\ref{boundm_{sep}}) is saturated.}
\begin{proof}
Let $\{\Pi_j\}_{j=1}^m=\{\Pi^s_j\}_{j=1}^{m_{sep}}\cup\{\Pi^e_j\}_{j=m_{sep}+1}^m$ be the tight measurement composed by $m_{sep}$ fully separable and $m-m_{sep}$ entangled ($k$-uniform) subnormalized rank-one projectors. The proof basically consists in taking partial trace to Eq.(\ref{rho}) over a given set of $k$ parties denoted $X_i$, for $i\in\{1,\dots,\binom{N}{k}\}$. Calculation of these reductions involve some fine details to be considered below one-by-one. Let us start by applying partial trace to the separable POVM projectors
\begin{eqnarray} 
    \mathrm{Tr}_{X_i}(\Pi_j^s) & =& \frac{d^N}{m}\mathrm{Tr}_{X_i}({|\varphi^s_j\rangle\langle\varphi^s_j|}) 
                                        = \frac{d^N}{m}{|\tilde{\varphi}^{(i)}_j\rangle\langle\tilde{\varphi}^{(i)}_j|}\nonumber \\ 
                                        &=& \frac{d^N}{m} \frac{m_{sep}}{d^k}\left(\frac{d^k}{m_{sep}}{|\tilde{\varphi}^{(i)}_j\rangle\langle\tilde{\varphi}^{(i)}_j|}\right) = \frac{m_{sep} d^N}{m d^k} \tilde{\Pi}^{(i)}_j,
\end{eqnarray}
where $\tilde{\Pi}^{(i)}_j=\frac{d^k}{m_{sep}}{|\tilde{\varphi}^{(i)}_j\rangle\langle\tilde{\varphi}^{(i)}_j|}$ are the suitable subnormalized rank-one POVM projectors according to definition of a tight measurement composed by $m_{sep}$ operators for the subset of $k$ parties $X_i$ (see Section \ref{S2}). For $k$-uniform projectors $\Pi_j^e$ we have uniform reductions to $k$ parties, that is
\begin{eqnarray}
    \mathrm{Tr}_{X_i}(\Pi_j^e)              & = &\frac{d^N}{m}\mathrm{Tr}_{X_i}({|\varphi^e_j\rangle\langle\varphi^e_j|}) =\frac{d^N}{m d^k}\, \mathbb{I}_{d^k}.
\end{eqnarray}
Also, for the identity operator we have
\begin{equation}
    \mathrm{Tr}_{X_i}(\mathbb{I}_{d^N})  =d^{N - k} \mathbb{I}_{d^k}.
\end{equation}
Another important detail to be taken into account is the suitable expression for probabilites. In order to simplify expressions we assume that $\rho=\tilde{\rho}\otimes\mathbb{I}/d^{N-k}$. This assumption can be done without loss of generality, as $\tilde{\rho}$ is an arbitrary density matrix acting on the subset of $k$ parties $X_i$. Therefore, probabilities associated to separable projectors $\Pi_j^s$ can be written as:
\begin{eqnarray}
p_j   & =& \mathrm{Tr}(\rho \Pi_j^s)
                                         = \frac{d^N}{m} \frac{d^k}{d^N}\mathrm{Tr}[(\rho_{X_i}\otimes\mathbb{I}_{d^{N-k}})({\Pi}^{(i)}_j\otimes{{\Pi}_j}_{\bar{X}_i})] \nonumber \\
                                        & =& \frac{d^N}{m} \frac{d^k}{d^N} \frac{m_{sep}}{d^k}\frac{d^k}{m_{sep}} \mathrm{Tr}(\rho_{X_i} {{\Pi}^{(i)}_j})  = \frac{m_{sep}}{m} \mathrm{Tr}(\rho_{X_i} \tilde{\Pi}^{(i)}_j)
                                         = \frac{m_{sep}}{m} \tilde{p}^{(i)}_j,
\end{eqnarray}
where we assumed the decomposition $\Pi_j^s=(d^N/m)\,{{\Pi}_j}_{X_i}\otimes{{\Pi}_j}_{\bar{X}_i}$, for $j\in\{1,\dots,m_{sep}\}$. Here, $\bar{X}_i$ denotes the subset of $N-k$ parties complementary to $X_i$. The remaining ingredient concerns to the simplified expression for a partial sum of probabilities. That is,
\begin{equation}
    \sum_{j=m_{sep} + 1}^m p_j=1 - \sum_{j=1}^{m_{sep}} p^{(i)}_j = 1 - \frac{m_{sep}}{m} \sum_{j=1}^{m_{sep}} \tilde{p}^{(i)}_j = 1 - \frac{m_{sep}}{m}, \label{last}
\end{equation}
where $\tilde{p}^{(i)}_j=\mathrm{Tr}(\tilde{\rho}\,\tilde{\Pi}^{(i)}_j)$ and condition $\sum_{j=1}^{m_{sep}}\tilde{p}^{(i)}_j=1$ is imposed because we require that the set of operators $\{\tilde{\Pi}^{(i)}_j\}$, $j\in\{1,\dots,m_{sep}\}$ forms a tight measurement for every subset of $k$ parties $X_i$, $i\in\{1,\dots,\binom{N}{k}\}$.

Taking into account all the ingredients described above we are now in position to apply partial trace to Eq.(\ref{rho}) with respect to the subset of $k$ parties $X_i$, thus obtaining
\begin{widetext}
\begin{align}\label{partialT}
    \mathrm{Tr}_{X_i}(\rho) & = \frac{m(d^N + 1)}{d^N} \left(\sum_{j=1}^{m_{sep}} p_j^s \mathrm{Tr}_{X_i}(\Pi_j) + \sum_{j=m_{sep} + 1}^{m} p_j \mathrm{Tr}_{X_i}(\Pi_j^k)\right) - \mathrm{Tr}_{X_i}(\mathbb{I}_{d^N}) \nonumber \\
                        & = \frac{m(d^N + 1)}{d^N} \left(\frac{m_{sep} d^N}{m d^k}\sum_{j=1}^{m_{sep}} p_j \tilde{\Pi}_j + \frac{d^N}{m d^k}\sum_{j=m_{sep} + 1}^{m} p_j \mathbb{I}_{d^k}\right) - d^{N - k} \mathbb{I}_{d^k} \nonumber \\
                        & = \frac{m(d^N + 1)}{d^N}\left(\frac{m_{sep} d^N}{m d^k} \frac{m_{sep}}{m}\sum_{j=1}^{m_{sep}} \tilde{p}_j \tilde{\Pi}_j + \frac{d^N}{m d^k}(1 - \frac{m_{sep}}{m}) \mathbb{I}_{d^k}\right) - d^{N - k} \mathbb{I}_{d^k} \nonumber \\
   & = \frac{m_{sep}^2(d^N + 1)}{m d^k} \sum_{j=1}^{m_{sep}} \tilde{p}_j \tilde{\Pi}_j + \frac{m - (d^N + 1) m_{sep}}{m d^k} \mathbb{I}_{d^k}.
\end{align}
\end{widetext}
From comparing Eqs. (\ref{rho}) and (\ref{partialT}) we obtain the following restrictions on the parameters:
\begin{equation}
    \frac{m_{sep}^2 (d^N + 1)}{m d^k} = \frac{m_{sep} (d^k + 1)}{d^k},
\end{equation} 
and
\begin{equation}
    \frac{m - (d^N + 1)m_{sep}}{m d^k} = 1.
\end{equation}
These two equations are equivalent and reduce to $m_{sep}(d^N + 1) = m (d^k + 1)$, which concludes the proof.
\end{proof}

\section{Two-qubit tight measurements composed by four separable vectors}\label{Ap2}

In this section we provide a numerical study of two-qubit tight measurement. Below we present a list of a two-qubit SIC-POVM, i.e. tight measurement having $m=16$ outcomes, such that 5 out of its 16 rank-one projectors are separable.

\begin{eqnarray}
v^{sep}_1=&(&\hspace{-0.2cm}1, 0)\otimes(1, 0)\nonumber\\
v^{sep}_2=&(&\hspace{-0.2cm}0.829450, -0.071288 + 0.554012 i)\otimes(0.538325, -0.738169 + 0.406584 i)\nonumber\\
v^{sep}_3=&(&\hspace{-0.2cm}0.558493, 0.656070 + 0.507598 i)\otimes(0.801919,  0.432837 + 0.411797 i)\nonumber\\
v^{sep}_4=&(&\hspace{-0.2cm}0.822550, -0.502837 +  0.265641 i)\otimes(0.543781, 0.835135 + 0.082763 i)\nonumber\\
v^{sep}_5=&(&\hspace{-0.2cm}0.543031,  0.153431 +  0.825576 i)\otimes(0.823280, -0.032621 - 0.566696 i)\nonumber\\[5pt]
v_6=&(&\hspace{-0.2cm}0.447615, -0.389125-0.446939 i, -0.454783+0.230108 i, -0.276499+0.335022 i)\nonumber\\
v_7=&(&\hspace{-0.2cm}0.446765, 0.604844 -0.274049 i, -0.291296-0.259105 i, 0.439704 -0.118884 i)\nonumber\\
v_8=&(&\hspace{-0.2cm}0.447075, 0.00858814 -0.303724 i, 0.241987 +0.435012 i, -0.638674-0.228265 i)\nonumber\\
v_9=&(&\hspace{-0.2cm}0.446722, 0.334228 +0.585485 i, 0.0484938 +0.315171 i, 0.385974 -0.308672 i)\nonumber\\
v_{10}=&(&\hspace{-0.2cm}0.446515, -0.00980095-0.028985 i, 0.0311724 -0.417279 i, -0.0431798-0.789133 i)\nonumber\\
v_{11}=&(&\hspace{-0.2cm}0.447508, 0.340026 +0.0366212 i, 0.422419 -0.587076 i, -0.394988+0.060542 i)\nonumber\\
v_{12}=&(&\hspace{-0.2cm}0.447241, -0.200413-0.735024 i, 0.340286 -0.312953 i, 0.012735 -0.0751886 i)\nonumber\\
v_{13}=&(&\hspace{-0.2cm}0.448257, -0.213794+0.454123 i, -0.24267-0.231547 i, -0.354814+0.555638 i)\nonumber\\
v_{14}=&(&\hspace{-0.2cm}0.447515, -0.205972-0.235433 i, -0.0725021-0.207019 i, 0.614807 +0.525144 i)\nonumber\\
v_{15}=&(&\hspace{-0.2cm}0.44708, -0.45748+0.428185 i, 0.470043 -0.350614 i, 0.237987 +0.0835461 i)\nonumber\\
v_{16}=&(&\hspace{-0.2cm}0.447179, -0.109299+0.192073 i, -0.833937-0.142906 i, 0.0951736 -0.162052 i)\nonumber
\end{eqnarray}
\begin{table}[h]
    \centering
    \begin{tabular}{c|c|c}
        \hspace{0.4cm}$m$ \hspace{0.3cm} & Welch bound & Weighted frame potencial ($F_2$)   \\ \hline
        16 & 25.600 & 25.600   \\
        17 & 28.900 & 28.914   \\
        18 & 32.400 & 32.414   \\
        19 & 36.100 & 36.101   \\
        20 & 40.000 & 40.000 
    \end{tabular}
    \caption{Numerical optimization of weighted frame potential $F_2$ for generalised measurements of two-qubit systems ($D=2^2$), having $16\leq m\leq20$ measurement outcomes (see Eq.(\ref{potential})). Lower bound for $F_2$, i.e. Welch bound (\ref{Welch}), can be achieved if and only if the quantum measurement is tight. Our study suggests that tight measurement may exist for any number of measurement outcomes $m\geq D^2$.}
    \label{Tabla1}
\end{table}
The smallest weighted frame potential $F_2$ achieved by the above solution is 25.600034, where 25.6 is the value of the Welch bound (\ref{Welch}). As a further numerical study, in Table \ref{Tabla1} we show that tight measurement seem to exist for any number of vectors $m\geq16$, where $m=16$ correspond to the minimal possible number of outcomes (SIC-POVM).


\newpage
\begin{thebibliography}{99}
\bibitem{S08} Scott A 2008 J. Phys. A 41 055308

\bibitem{R05} Renes J 2005 Quantum Inf. Comput. 5 81

\bibitem{SWS07} Scott A, Walgate and J Sanders B 2007 Quantum Inf. Comput. 7 243

\bibitem{S06} Scott A 2006 J. Phys. A 39 13507

\bibitem{RS07} Roy A and Scott A 2007 J. Math. Phys. 48 072110

\bibitem{N81} Neumaier A 1981, \emph{Combinatorial configurations in terms of distances}, Dept. of Mathematics Memorandum 81-09 (Eindhoven University of Technology).

\bibitem{GA16} Graydon M and Appleby D 2016 J. Phys. A: Math. Theor. 49 085301

\bibitem{Z99} Zauner G Ph.D. thesis, University of Vienna, (1999)

\bibitem{DCEL09} Dankert C, Cleve R, Emerson J and Livine E 2009 Phys. Rev. A 80 012304

\bibitem{GAE07} Gross D, Audenaert K and Eisert J 2007 J. Math. Phys. 48 052104

\bibitem{CLM15} Chen B, Li T and Ming Fei S 2015 Quant. Inform. Process. 14 2281

\bibitem{XZ16} Xi Y and Jun Zheng Z 2016 Quant. Inform. Process. 15  5119

\bibitem{AE07} Ambainis A and Emerson J, Quantum t-designs: t-wise independence in the quantum world, Proceedings of the Twenty-Second Annual IEEE Conference on Computational Complexity, June 13-16, Washington, U.S.A. (2007). 

\bibitem{KR05} Klappenecker A and Roetteler M, Mutually Unbiased Bases are Complex Projective 2-Designs, Proceedings of the IEEE International Symposium on Information Theory, Adelaide, Australia, September 2005, p. 1740. Available online: arXiv:quant-ph/0502031

\bibitem{LNGGNDVS11} Lima G, Neves L, Guzm\'{a}n R,  G\'{o}mez E, Nogueira W, Delgado A, Vargas A and Saavedra C 2011 Opt. Express 19 3542

\bibitem{ECGNSXL13} Etcheverry S, Ca\~{n}as G, G\'{o}mez E, Nogueira W, Saavedra C,  Xavier G and Lima G 2013 Sci. Rep. 3 2316

\bibitem{B15} Bent N, Qassim H, Tahir A, Sych D, Leuchs G, L S\'{a}nchez-Soto, Karimi E and Boyd R 2015 Phys. Rev. X, 5 041006

\bibitem{DBBV17} Dall'Arno M, Brandsen S, Buscemi F and Vedral V 2017 Phys. Rev. Lett. 118 250501

\bibitem{HSFHWUZ15} Herbst T, Scheidl T, Fink M, Handsteiner J, Wittmann B, Ursin R and Zeilinger A 2015 PNAS 112(46) 14202

\bibitem{BF03} Benedetto J and Fickus M 2003 Adv. Comput. Math. 18 357

\bibitem{B08} Belovs A MSc thesis, University of Waterloo, Canada (2008) 

\bibitem{W03} Waldron S 2003 IEEE Trans. Inform. Theory 49 2307.

\bibitem{RBSC04} Renes J, Blume-Kohout R, Scott A and Caves C 2004 J. Math. Phys 45 2171


\bibitem{ACFW} Appleby M, Chien T, Flammia S and Waldron S 2017 Constructing exact symmetric informationally complete measurements from numerical solutions Preprint  arXiv:1703.05981

\bibitem{Marcus} Appleby D 2005 J. Math. Phys. 46 052107
  
\bibitem{Markus1} Grassl M 2009 On SIC-POVMs and MUBs in dimension $6$ Preprint quant-ph/0406175.

\bibitem{GS17}  Grassl M and Scott A 2017 J. Math. Phys. 58 12 122201

\bibitem{Scott} Scott A and Grassl A 2010 J. Math. Phys.  51 042203

\bibitem{Andrew} Scott A 2017 SICs: Extending the list of solutions Preprint arXiv:1703.03993

\bibitem{FHS} Fuchs C, Hoan M and Stacey B 2017 Axioms 6 3 21

\bibitem{I81} Ivanovic I 1981 J. Phys. A 14 3241

\bibitem{WF89} Wootters W and Fields B 1989 Ann. Phys. 191 363

\bibitem{S16} Skowronek L 2016   J.Math. Phys. 57 112201  

\bibitem{CGGZ18}  Czartowski J, Grassl M, Goyeneche D and \.{Z}yczkowski K 2018 (unpublished)

\bibitem{L97}  Leonhardt U 1997 Measurement Science and Technology 11 12

\bibitem{ZE11} Zhu H and Englert B 2011 Phys. Rev. A 84 022327

\bibitem{S04} Scott A 2004 Phys. Rev. A 69 052330

\bibitem{AC13} Arnaud L and Cerf N 2013 Phys. Rev. A 87 012319

\bibitem{L78} Lubkin E 1978 J. Math Phys. 19 1028

\bibitem{WPZ11} Wie\'{s}niak M, Paterek T and Zeilinger A 2011 New J. Phys. 13 053047

\bibitem{L11} Lawrence J 2011 Phys. Rev. A 84 022338

\bibitem{HCRLL12} Helwig W, Cui W, Riera A, Latorre J and Lo H 2012 Phys. Rev. A 86, 052335

\bibitem{HS00} Higuchi A amd Sudbery A 2000 Phys. Lett. A 272 213

\bibitem{H98} Hoggar S 1998  Geometriae Dedicata 69, 287--289

\bibitem{Z12} Zhu H Ph.D. thesis, National University of Singapore (2012)

\bibitem{JW15}  Jedwab J  and  Wiebe A, A  simple  construction  of complex equiangular lines, in Algebraic Design Theory and Hadamard Matrices, (C. J. Colbourn, ed. Springer, 2015)  and Preprint arXiv:1408.2492

\bibitem{CKW00} Coffman V, Kundu J and Wootters W 2000 Phys. Rev. A 61 052306

\bibitem{BBRV02} Bandyopadhyay S, Boykin P, Roychowdhury V and Vatan F 2002 Algorithmica 34, 4, 512--528 (2002)


\bibitem{G17} M. Grassl, private communication, 26-10-2017

\end{thebibliography}
\end{document}